\documentclass[11pt]{article}
\usepackage{amsmath,amsfonts,amsthm,amssymb,color}
\usepackage{thm-restate}
\usepackage[colorlinks=true]{hyperref}
\usepackage{fancybox}

\newtheorem{theorem}{Theorem}[section]
\newtheorem{corollary}[theorem]{Corollary}
\newtheorem{lemma}[theorem]{Lemma}

\newtheorem{fact}[theorem]{Fact}

\theoremstyle{definition}
\newtheorem{definition}[theorem]{Definition}

\newenvironment{fminipage}%
  {\begin{Sbox}\begin{minipage}}%
  {\end{minipage}\end{Sbox}\fbox{\TheSbox}}

\newenvironment{algbox}[0]{\vskip 0.2in
\noindent 
\begin{fminipage}{6.3in}
}{
\end{fminipage}
\vskip 0.2in
}

\def\expec#1#2{{\mathbb{E}}_{#1}\left[ #2 \right]}

\def\defeq{\stackrel{\mathrm{def}}{=}}

\def\abs#1{\left|#1  \right|}

\def\norm#1{\left\| #1 \right\|}

\newcommand\PPi{\boldsymbol{\Pi}}

\newcommand\ssigma{\boldsymbol{\sigma}}

\newcommand\ttau{\boldsymbol{\tau}}

\def\aa{\pmb{\mathit{a}}}
\newcommand\bb{\boldsymbol{\mathit{b}}}
\newcommand\cc{\boldsymbol{\mathit{c}}}

\newcommand\ee{\boldsymbol{\mathit{e}}}

\newcommand\pp{\boldsymbol{\mathit{p}}}

\newcommand\rr{\boldsymbol{\mathit{r}}}
\renewcommand\ss{\boldsymbol{\mathit{s}}}

\newcommand\uu{\boldsymbol{\mathit{u}}}
\newcommand\vv{\boldsymbol{\mathit{v}}}
\newcommand\ww{\boldsymbol{\mathit{w}}}
\newcommand\yy{\boldsymbol{\mathit{y}}}

\newcommand\xx{\boldsymbol{\mathit{x}}}
\newcommand\wwhat{\boldsymbol{\widehat{\mathit{w}}}}

\newcommand\wwbar{\overline{\boldsymbol{\mathit{w}}}}

\renewcommand\AA{\boldsymbol{\mathit{A}}}
\newcommand\BB{\boldsymbol{\mathit{B}}}
\newcommand\CC{\boldsymbol{\mathit{C}}}

\newcommand\EE{\boldsymbol{\mathit{E}}}

\newcommand\II{\boldsymbol{\mathit{I}}}

\newcommand\MM{\boldsymbol{\mathit{M}}}

\newcommand\PP{\boldsymbol{\mathit{P}}}
\newcommand\QQ{\boldsymbol{\mathit{Q}}}

\newcommand\RR{\boldsymbol{\mathit{R}}}
\renewcommand\SS{\boldsymbol{\mathit{S}}}

\newcommand\WW{\boldsymbol{\mathit{W}}}
\newcommand\WWbar{\overline{\boldsymbol{\mathit{W}}}}
\newcommand\VV{\boldsymbol{\mathit{V}}}

\newcommand{\expct}[2]{\ensuremath{\mathop{\text{\normalfont \textbf{E}}}_{#1}}\left[#2\right]}

\newcommand\trace[1]{{\boldsymbol{\mathit{tr}}}\left[ #1 \right] }

\newcommand\nnz{\boldsymbol{nnz}}
\newcommand\poly{\mathit{poly}}
\renewcommand{\Re}{\mathbb{R}}

\begin{document}

\title{$\ell_p$ Row Sampling by Lewis Weights}

\author{
  Michael B. Cohen\\
  M.I.T.\\
  \texttt{micohen@mit.edu}
  \and
 Richard Peng\\
  M.I.T.\\
  \texttt{rpeng@mit.edu}
}

\maketitle


\begin{abstract}
We give a simple algorithm to efficiently sample the rows of a matrix while
preserving the p-norms of its product with vectors.
Given an $n$-by-$d$ matrix $\AA$, we find with high probability and in
input sparsity time an $\AA'$ consisting of about $d \log{d}$ rescaled rows of
$\AA$ such that $\norm{\AA \xx}_1$ is close to $\norm{\AA' \xx}_1$ for all vectors $\xx$.
We also show similar results for all $\ell_p$ that give nearly optimal sample bounds 
in input sparsity time.  Our results are based on sampling by ``Lewis weights'',
which can be viewed as statistical leverage scores of a reweighted matrix.
We also give an elementary proof of the guarantees of this sampling process for $\ell_1$.
\end{abstract}

\pagenumbering{gobble}

\pagebreak

\pagenumbering{arabic}


\section{Introduction}
\label{sec:intro}
Randomized sampling is an important tool in the design of efficient algorithms.
A random subset often preserves key properties of the entire data set,
allowing one to run algorithms on a small sample.
A particularly useful instance of this phenomenon is row sampling of matrices.
For a $n \times d$ matrix $\AA$ where $n \gg d$ and any error parameter
$\epsilon > 0$, we can find $\AA'$ with a few (rescaled) rows of $\AA$ such that
\begin{equation*}
\norm{\AA \xx}_{p} \approx_{1 + \epsilon} \norm{\AA' \xx}_{p}
\end{equation*}
for all vectors $\xx \in \Re^{d}$. Here $\approx_{1 + \epsilon}$
denotes a multiplicative error between $(1 + \epsilon)^{-1}$ and $(1 + \epsilon)$.

Originally studied in statistics~\cite{Talagrand90,RudelsonV07,Tropp12},
the row sampling problem has received much attention recently in randomized numerical linear
algebra~\cite{DrineasMM06,DrineasMMW11,ClarksonW13,MahoneyM13,NelsonN12,LiMP13,CohenLMMPS14:arxiv}
and in graph algorithms as spectral sparsification~\cite{SpielmanS08:journal,BatsonSS12,KoutisLP12}.
These works led to a good understanding of row sampling for the $p = 2$ case.
If the rows of $\AA$ are sampled with probabilities proportional to their
statistical leverage scores, matrix Chernoff bounds~\cite{AhlswedeW02,RudelsonV07,Tropp12}
state that $\AA'$ with $O(d \log{d} / \epsilon^{-2})$ is a good approximation with high probability. 
Recently, Clarkson and Woodruff developed oblivious subspace embeddings
that brought the runtime of these algorithms down to input-sparsity time~\cite{ClarksonW13}.
A derandomized procedure by Batson et al.~\cite{BatsonSS12} can also
reduce the number of rows in $\AA'$ to $O(d / \epsilon^2)$, at the cost of a worse,
but still polynomial, runtime.

Substantial progress has also been made for other values of $p$~\cite{DasguptaDHKM09,SohlerW11,ClarksonDMMMW13,MahoneyM13},
leading to input-sparsity time algorithms that return samples
with about $d^{2.5}$ rows when $1 \leq p \leq 2$.
The $p = 1$ case is of particular interest due to its relation
to robust regression, or $\ell_1$-regression~\cite{Candes06}.
In this setting, the existence of samples of size $O(d \log{d})$ was
shown by Talagrand~\cite{Talagrand90} using Banach space theory.

Talagrand's result, and the best known sampling results for general
$\ell_p$ norms~\cite{BourgainLM89,TalagrandLp}, are based on a
``change of density'' construction originally due to Lewis~\cite{lewis}. This construction
assigns a weight, analogous to a leverage score, to each row; these can be used
directly as sampling probabilities.  We will refer to these weights as ``Lewis weights''.
Given their direct use as sampling probabilities, the primary algorithmic challenge
is to be able to compute, or at least approximate, these weights.  This paper provides
the first input-sparsity time algorithms, and in fact the first polynomial time algorithms,
for this problem.  That in turn leads to the first polynomial time algorithmic versions
of the Talagrand and other Lewis weight-based results.

In particular, we give a simple iterative algorithm that approximates Lewis weights
through repeated computations of statistical leverage scores.
Sampling by these approximate weights then leads to the following result:

\begin{theorem}
\label{thm:main}
Given a matrix $\AA$ and any error parameter $\epsilon > 0$,
there is a distribution of matrices $\SS$ with $O(d \log d \epsilon^{-2})$
rows and one nonzero entry per row such that with high probability
\begin{equation*}
\norm{\SS \AA \xx}_1 \approx_{1 + \epsilon} \norm{\AA \xx}_1
\qquad \text{for all } \xx \in \Re_{d}.
\end{equation*}
Furthermore, we can sample from this distribution using $O(\log \log{n})$ calls
to computing 2-approximate statistical leverage scores of matrices of the
form $\WW \AA$ where $\WW$ is a non-negative diagonal matrix.
\end{theorem}
This routine can be combined with any algorithm for computing (approximate) statistical leverage scores.
By invoking input-sparsity time routines~\cite{ClarksonW13,MahoneyM13,
NelsonN12,LiMP13,CohenLMMPS14:arxiv} and approximating intermediate
objects to a coarser granularity, we obtain algorithms that compute $\AA'$ with
$O(d \log{d} \epsilon^{-2})$ rows in $O(\nnz(\AA) + d^{\omega + \theta} )$ time.
Here $\omega$ is the matrix-multiplication exponent and $\theta > 0$ is any constant.
This is a substantial improvement over previous results.
The previous best algorithm with a comparable running time gives a sample size of
about $d^{3.66}$~\cite{LiMP13}; algorithms for obtaining samples with
$O(d^{2.5})$ rows rely on an expensive ellipsoidal rounding algorithm~\cite{DasguptaDHKM09},
and therefore have prohibitively large $\poly(d)$ terms.

We also give input-sparsity time algorithms for sampling for general $\ell_p$ norms.
However, the sample count is slightly higher for $p \in (1,2)$ and substantially higher
(around $d^{p/2}$) for $p > 2$.  Furthermore, the runtime gets substantially worse as
$p$ approaches and exceeds 4 (we  use the ellipsoid algorithm in this regime).
Nonetheless, these algorithms do much better than previous ones: they give a near-optimal
sample count, and their dependence on $d$ (though high) is of the form
$d^{p/2 (1 + \theta) + C}$
(better than could be obtained in approaches based on resparsifying the
outputs of existing algorithms).

A table summarizing our main algorithmic results is in Figure~\ref{fig:summary}.
The bounds omit logarithmic factors in the runtime, $d^{\theta}$ tradeoffs
with input sparsity terms, success probabilities, and $\epsilon$
and constant factor $p$ dependences.
All of these algorithms can be run with an $\nnz(\AA) \log n$ term and also
support a tradeoff to get pure input sparsity time.

\begin{figure}

\begin{center}

\begin{tabular}{| l | l | l | l | }
\hline
$p$ & polynomial runtime term & Sample count & Relevant Sections\\ \hline
$p = 1$ & $d^{\omega}$ & $d \log d$ &
Sections~\ref{sec:overview},~\ref{sec:compute},~\ref{sec:ellp},~and~\ref{sec:concentration}\\ \hline
$1 < p < 2$ & $d^{\omega}$ & $d \log d \log \log^2 d$ & Sections~\ref{sec:overview},~\ref{sec:compute}, and~\ref{sec:ellp}\\ \hline
$2 < p < 4$ & $d^{\omega}$ & $d^{p/2} \log d$ &
Sections~\ref{sec:overview},~\ref{sec:compute},~and~\ref{sec:ellp}\\ \hline
$2 \leq p $ & $d^{p/2+C}$ & $d^{p/2} \log d$ &
Sections~\ref{sec:overview},~\ref{sec:convex},~\ref{sec:sparseconvex}~and~\ref{sec:ellp}\\ \hline
\end{tabular}

\end{center}

\caption{Summary of Our Algorithmic Results}
\label{fig:summary}
\end{figure}

In Section~\ref{sec:ellp}, we briefly go over known statistical results~\cite{Talagrand90,
TalagrandLp,BourgainLM89} on the concentration of sampling by Lewis weights.
One issue with invoking these results is that the bounds in them do not explicitly
describe non-uniform sampling based on upper bounds of Lewis weights.
Instead, the more common form of a main theorem in these results is that
if the Lewis weights are uniformly small, uniformly sampling half the rows gives
a good approximation.
We use standard techniques (which we describe in more detail in
Appendix~\ref{sec:reduct}) to show that the latter implies the former.

Additionally, in the case of $\ell_1$, we give an alternate, elementary proof
of the concentration of the sampling process in Section~\ref{sec:concentration}.
It simplifies many components of previous
proofs~\cite{Talagrand90,Pisier99:book} while following the same basic strategy.
The bound is slightly weaker in some artificial parameter ranges, but is also stronger
in giving much sharper tails.  It is arguably even simpler than proofs of $\ell_2$ matrix Chernoff
bounds~\cite{Vershynin09,Harvey11}.  Unfortunately, the picture for general $\ell_p$ seems
to be much more difficult: published proofs for $p \not \in \{ 1, 2 \}$ use deep results from
Gaussian processes.  We consider it an open problem to find elementary proofs for these ranges.

Our results significantly improve algorithms for the well studied
$\ell_p$, and specifically $\ell_1$, row sampling problem.
We also give a version of $\ell_1$ matrix concentration bound that's
analogous to the widely-used $\ell_2$ matrix concentration bounds.
We believe these results and the simplicity of our techniques
show the usefulness of Lewis weights as an algorithmic tool
for randomized numerical linear algebra in $p$-norms.

The paper is structured as follows: in Section~\ref{sec:overview} we formalize
the matrix row sampling problem and give an overview of our result.
Our simple iterative algorithm for computing approximate Lewis weights is in Section~\ref{sec:compute};
this applies to all $p < 4$.
An alternative approach to computing Lewis weights based on convex optimization, valid for all $p \geq 2$,
is in Section~\ref{sec:convex}.
Section~\ref{sec:properties} gives properties of Lewis weights useful in some of our arguments.
An input-sparsity time algorithm using the convex optimization approach is in Section~\ref{sec:sparseconvex}.
Section~\ref{sec:ellp} summarizes the known sampling results for each range of $p$, and describes
a way to take existing proofs and obtain analyses of our simple sampling procedure.
Section~\ref{sec:concentration} gives most of our elementary proof of the validity
of $\ell_1$ sampling by Lewis weight.
Finally, Appendix~\ref{sec:properties-proofs} proves the properties from Section~\ref{sec:properties},
and Appendix~\ref{sec:reduct} gives the proof of the reduction from our sampling procedure
given in Section~\ref{sec:ellp}.


\section{Background and Overview}
\label{sec:overview}

The paper deals extensively with vectors and matrices.
The $\ell_p$-norm of a vector $\xx \in \Re^{d}$ is defined as
\begin{equation*}
\norm{\xx}_p = \left( \sum_{i = 1}^{d} \left|\xx_i\right|^{p} \right)^{1/p}. 
\end{equation*}
For a matrix $\AA$, we will use $n$ and $d$ to denote its number of
rows and columns respectively, and $\aa_i$ to denote the vector
corresponding to the $i\textsuperscript{th}$ row of $\AA$.
Note that $\aa_i$ is a {\em column} vector.
The $\ell_p$-norm of a vector $\xx$ w.r.t. $\AA$ can then be written as
\begin{equation*}
\norm{\AA \xx}_p = \left( \sum_{i = 1}^{n} \left| \aa_i^T \xx \right|^{p} \right)^{1/p}.
\end{equation*}
We will also assume our matrices are full rank, since otherwise we can
either project onto its rank space, or use pseudo-inverses accordingly.

Most of our analyses revolve around multiplicative errors.
Here we follow the approximation notation from~\cite{CohenLMMPS14:arxiv}.
For a  parameter $\alpha \geq 1$, we say two quantities $x$ and $y$ satisfy
$x \approx_{\alpha} y$ if
\begin{equation*}
\frac{1}{\alpha} x \leq y \leq \alpha x.
\end{equation*}
Note that if $x \approx_{\alpha} y$, then for any power $p$ we have
$x^{p} \approx_{\alpha^{|p|}} y^{p}$.

A cruciual definition in $\ell_2$ row sampling and matrix concentration bounds is
statistical leverage scores.
The statistical leverage score of a row $\aa_i$ is
\begin{equation*}
\ttau_{i}\left(\AA\right)
\defeq \aa_i^T \left( \AA^T \AA \right)^{-1} \aa_i
= \norm{\left( \AA^T \AA \right)^{-1/2} \aa_i}_2^2.
\end{equation*}
It can be viewed as the squared norm of the $i\textsuperscript{th}$ row after the
statistical whitening transform (see e.g.. Hyvarinen et al.~\cite{HyvarinenO00}).
Equivalently, it is also often defined as the squared row norm of the matrix of left singular vectors
of $\AA$.
The following facts about statistical leverage scores underpin their role
in $\ell_2$ row sampling.
\begin{fact}
\label{fact:tau}
\begin{enumerate}
\item \label{item:foster} (Foster's theorem~\cite{Foster53}) $\sum_{i = 1}^{n} \ttau_{i} (\AA) \leq d$,
\item \label{item:upper1} $\ttau_{i}(\AA) \leq 1$.
\item \label{item:l2concentration} ($\ell_2$ Matrix Concentration Bound)
There exists an absolute constant $C_s$ such that for any matrix $\AA$ and any
set of sampling values $\pp_i$ satisfying
\begin{equation*}
\pp_i \geq C_s \ttau_{i} \left( \AA \right) \log{d} \epsilon^{-2},
\end{equation*}
if we generate a matrix $\SS$ with $N = \sum \pp_i$ rows, each chosen
independently as the $i$\textsuperscript{th} basis vector,
times $\frac{1}{\sqrt{\pp_i}}$, with probability $\frac{\pp_i}{N}$
then with high probability we will have
$\norm{\SS \AA \xx}_2 \approx_{1 + \epsilon} \norm{\AA \xx}_2$ for all vectors $\xx$.
\end{enumerate}
\end{fact}
Combining part~\ref{item:foster} and part~\ref{item:l2concentration} immediately implies that replacing
$\AA$ with $O(d \log d / \epsilon^2)$ reweighted row samples from $\AA$ can give a $(1+\epsilon)$ approximation.

$\ell_p$ spaces are more complicated and lack many useful properties of $\ell_2$.
The Lewis weight approach can be seen as a particular way to tap into the niceness of $\ell_2$
by defining, for any matrix $\AA$, a corresponding matrix $\BB$ so that $\norm{\AA \xx}_p$
is in some sense related to $\norm{\BB x}_2$.  One conceivable notion of relatedness would be
to simply minimize the maximum distortion between the norms.  This would define $\BB$ based on
the John ellipsoid for the convex body $\norm{\AA \xx}_p \leq 1$; this is essentially the technique
used in \cite{DasguptaDHKM09} and follow-up works.  However, it does not lead to tight bounds, and
the Lewis approach is different (although Section~\ref{sec:convex} reveals a similarity).

The most naive approach would seem to be to simply set $\BB = \AA$.  Here, however, one
sees that $\BB$ cannot properly capture $\AA$, since $\norm{\BB x}_2$ is not invariant under
``change of density.''  For instance, one may ``split'' a row $\aa_i$ in $\AA$ into $k$ pieces,
each equal to $k^{-1/p} \aa_i$, and $\norm{\AA \xx}_p$ will remain unchanged.
$\norm{\AA \xx}_2$, though, will change, and could be arbitrarily distorted by subdividing different
rows different amounts.

Instead, we adapt this naive approach by using a different ``density.''  We effectively assume
that a given row $\aa_i$ really represents $\ww_i$ rows, each equal to $\ww_i^{-1/p} \aa_i$.
$\ww_i$ may not be an integer, but this does not really matter; they could be viewed as
weights in a weighted $\ell_p$ norm rather than a number of copies.  When switching to $\ell_2$,
the original row $i$ of $\AA$ would still correspond to $\ww_i$ rows equal to $\ww_i^{-1/p} \aa_i$, but
this now has the same effect as one row equal to $\ww_i^{1/2 - 1/p} \aa_i$, rather than simply $\aa_i$
itself.  Putting this together, we will define $\BB = \WW^{1/2 - 1/p} \AA$.

This still leaves the question of how to actually choose the weights $\ww_i$ (the
specific change of density).  Intuitively, we want the split up rows, $\ww_i^{-1/p} \aa_i$
to be normalized in some way.  Lewis's change of density gives a simple and natural
notion of this: each of these normalized rows should have leverage score 1
(defined in terms of a $\ww_i$-weighted $\ell_2$ norm), or, more explicitly,
the $i$\textsuperscript{th} row of $\BB$ should end up with leverage score $\ww_i$.  Note that
this is a somewhat circular characterization: $\ww_i$ must match the leverage
scores of $\BB$, but $\BB$ itself depends on $\ww_i$.  The Gram matrix
of $\BB$ is $\AA^T \WW^{1 - 2/p} \AA$.  Writing this all in terms of the
original matrix $\AA$ gives:
\begin{definition}
\label{def:lewis}
For a matrix $\AA$ and norm $p$, the $\ell_p$ Lewis weights $\wwbar$ are the unique
weights such that for each row $i$ we have
\begin{equation*}
 \wwbar_i = \ttau_{i} \left( \WWbar^{1/2 - 1/p} \AA \right).
\end{equation*}
 or equivalently
\begin{equation*}
\aa_i^T \left( \AA^T \WWbar^{1 - 2/p} \AA \right)^{-1} \aa_i = \wwbar^{2/p}_i.
\end{equation*}
\end{definition}
We will make extensive use of the second formulation
since it groups the $\wwbar_i$ values on one side.
It also involves measuring the operator $\AA^T \WWbar^{1 - 2/p} \AA$ against a fixed
vector $\aa_i$, and therefore allows us to incorporate bounds on the operator.
These definitions have analogs in recent speedups of interior point methods
by Lee and Sidford~\cite{LeeS14}; for example, the requirement on the weight function
in the LP algorithm of~\cite{LeeS13a:arxiv}.

Note that for the case where $p = 2$, $\WWbar^{1/2 - 1/p}$ is the identity matrix,
so the Lewis weights are just the leverage scores.  In the case of general $p$,
it is not immediately clear that the Lewis weights must actually exist or be unique
because of the apparently circularity; existence and uniqueness was first established
by Lewis in~\cite{lewis}.  This paper also includes proofs of the existence and
uniqueness of Lewis weights, which follow directly from the arguments used to
prove our algorithms.  The first is in Section~\ref{sec:compute} (Corollary~\ref{cor:itlewis}
and applies to $p < 4$.  The second is in Section~\ref{sec:convex} (Corollary~\ref{cor:conlewis})
and gives a proof of existence for all $p$ and uniquness for $p \geq 2$ (so together, these
proofs give existence and uniqueness for all $p$); the resulting proof from that section
is a restatement of traditional proofs.

In Section~\ref{sec:concentration} (combined with Section~\ref{sec:ellp}),
we will prove the following concentration bound,
which is a variant of Proposition 2 from~\cite{Talagrand90}.
\begin{restatable}[$\ell_1$ Matrix Concentration Bound]{theorem}{lonechernoff}
\label{thm:l1chernoff}
There is an absolute constant $C_s$ such that given a matrix $\AA$
with $\ell_1$ Lewis weights $\wwbar$, for any set of sampling values $\pp_i$,
$\sum_i \pp_i = N$,
\begin{equation*}
\pp_i \geq C_s \wwbar_i \log(N) \epsilon^{-2},
\end{equation*}
if we generate a matrix $\SS$ with $N$ rows, each chosen
independently as the $i$\textsuperscript{th} standard basis vector,
times $\frac{1}{\pp_i}$, with probability $\frac{\pp_i}{N}$,
then with high probability we have
\begin{equation*}
\norm{\SS \AA \xx}_1
\approx_{1 + \epsilon} \norm{\AA \xx}_1
\end{equation*}
for all vectors $\xx \in \Re^{d}$.  In particular, with constant factor
approximations to the Lewis weights, $O(d \log(d/\epsilon) \epsilon^{-2})$
row samples suffices.
\end{restatable}

The bound on sample count also follows from Fact~\ref{fact:tau}.
As a result, it remains to compute approximate $\ell_p$ Lewis weights.
We present a novel, but extremely simple, iterative scheme that can do this for all $p < 4$.
We repeatedly perform
\begin{equation*}
\ww'_i \leftarrow \left( \aa_i^T \left( \AA^T \WW^{1 - 2/p} \AA \right)^{-1} \aa_i \right)^{p/2}
\end{equation*}
using a possibly approximate algorithm for computing statistical leverage scores.
This routine resembles the use of $\ell_2$-regression to solve $\ell_1$
regression through repeated reweighting~\cite{ChinMMP13}.
Its convergence also mimics the reduction in row counts in the iterative $\ell_p$
row sampling algorithm by Li et al.~\cite{LiMP13}.
We will prove the following Lemma regarding this procedure in Section~\ref{sec:compute}
\begin{lemma}
\label{lem:computeLewis}
For any fixed $p < 4$, given a routine \textsc{ApproxLeverageScores} for computing
$\beta$-approximate statistical leverage scores of rows of matrices of the form $\WW \AA$ for
$\beta = n^{\Omega(\theta)}$, we can compute a $n^{\theta}$ approximation to
$\ell_p$-Lewis weights for $\AA$ with $O(\log(\theta))$ calls to \textsc{ApproxLeverageScores}. 
\end{lemma}

Combining this Lemma with Theorem~\ref{thm:l1chernoff} or \cite{Talagrand90}, using $\theta$ proportional
to $\frac{1}{\log n}$, gives an algorithm satisfying the conditions of Theorem~\ref{thm:main}.
We may use standard techniques, as described, for instance, in \cite{CohenLMMPS14:arxiv}, to
get an input sparsity time algorithm as claimed in the introduction.
The idea is to run in two phases, first using a constant $\theta$, then resparsifying with $\theta$
proportional to $\frac{1}{\log n}$ (getting a $d^{\theta}$ dependence, rather than just $n^\theta$,
turns out to be automatic, since if $n > d^4$, the $\nnz$ term will dominate anyway).
The same process can be used for any $p \leq 4$, with sparsifier
sizes as stated in Section~\ref{sec:ellp}; however, due to the size of the intermediate sparsifier,
the polynomial dependence will be $d^{\max(\omega, p/2+1) + \theta}$.

It is worth noting that the technique in Section~\ref{sec:sparseconvex}, can also be be combined with
these methods (with Lemma~\ref{lem:computeLewis} replacing the convex optimization technique of
Theorem~\ref{thm:polyslow} from Section~\ref{sec:convex}) to give input sparsity time algorithms.
The arguments in Section~\ref{sec:sparseconvex} allow better parameters in the ranges $p \leq 2$
and $p \leq 4$ (this is discussed in more detail there), and this algorithm can support a
slightly better tradeoff between the input sparsity cost and the polynomial term.

\section{Iteratively Computing Approximate Lewis Weights}
\label{sec:compute}

Pseudocode of our algorithm is given in Figure~\ref{fig:iteration}.

\begin{figure}[ht]

\begin{algbox}

$\ww = \textsc{LewisIterate}(\AA, p, \beta, \ww)$

\begin{enumerate}
	\item For $i = 1 \ldots n$
	\begin{enumerate}
		\item Let $\tilde{\ttau_i} \approx_{\beta} \ttau_{i} \left( \WW^{1/2 - 1/p} \AA \right)$ be a $\beta$-approximation of the statistical leverage score of row $i$ in $\WW^{1/2 - 1/p} \AA$.
		\item Set $\wwhat_i \leftarrow \left( \ww_i^{2/p - 1} \tilde{\ttau}_i \right)^{p/2} \approx_{\beta^{p/2}} (\aa_i^T \left( \AA^T \WW^{1 - 2/p} \AA \right)^{-1} \aa_i)^{p/2}$.
	\end{enumerate}
	\item Return $\wwhat$.
\end{enumerate}

$\ww =\textsc{ApproxLewisWeights}(\AA, p, \beta, T)$

\begin{enumerate}
	\item Initialize $\ww_i  = 1$
	\item For $t = 1 \ldots T$
	\begin{enumerate}
		\item Set $\ww \leftarrow \textsc{LewisIterate}(\AA, p, \beta, \ww)$.
	\end{enumerate}
	\item Return $\ww$.
\end{enumerate}
\end{algbox}

\caption{Iterative Algorithm for Computing Lewis Weights}
\label{fig:iteration}

\end{figure}

In our proof, we make use of the generalization of approximations
 to the matrix setting via the Loewner partial ordering.
For two symmetric matrices $\PP$ and $\QQ$, we have $\PP \preceq \QQ$
if $\PP - \QQ$  is positive semi-definite.
We will then use $\PP \approx_{\alpha} \QQ$ to denote
\begin{equation*}
\frac{1}{\alpha} \PP \preceq \QQ \preceq \alpha \PP. 
\end{equation*}
Note that this is equivalent to $\xx^T \PP \xx \approx_{\alpha} \xx^T \QQ \xx$ for
all vectors $\xx$.
Two facts of particular importance for spectral approximation of matrices
are its composition and inversion.
\begin{fact}
\label{fact:operatorApprox}
\begin{enumerate}
\item If $\PP \approx_{\alpha} \QQ$, then for any matrix $\AA$, we also have
$\AA^T \PP \AA \approx_{\alpha} \AA^T \QQ \AA$.
\item If $\PP \approx_{\alpha} \QQ$, then $\PP^{-1} \approx_{\alpha} \QQ^{-1}$.
\end{enumerate}
\end{fact}

First, we show that (for $p < 4$) $\textsc{LewisIterate}(\AA, p, 1, _)$ acts as a
contraction mapping with respect to the $\ell_\infty$ distance of the $\log$ weights:
\begin{lemma}
\label{lem:contraction}
Given $\AA$, $p$, and weight sets $\vv$ and $\ww$ such that $\vv \approx_{\alpha} \ww$,
\begin{equation*}
\textsc{LewisIterate}(\AA, p, 1, \vv) \approx_{\alpha^{|p/2 - 1|}} \textsc{LewisIterate}(\AA, p, 1, \ww).
\end{equation*}
\end{lemma}
\begin{proof}
Since $\vv \approx_{\alpha} \ww$, we have, in the matrix setting
\begin{equation*}
\VV^{1 - 2/p} \approx_{\alpha^{|1 - 2/p|}} \WW^{1 - 2/p}.
\end{equation*}
Applying both items of Fact~\ref{fact:operatorApprox} we then have
\begin{equation*}
\left( \AA^T \VV^{1 - 2/p} \AA \right)^{-1} \approx_{\alpha^{|1-2/p|}} \left( \AA^T \WW^{1 - 2/p} \AA \right)^{-1}.
\end{equation*}
Then applying the definition of operator approximation, we have, for all $i$,
\begin{equation*}
\aa_i^T \left( \AA^T \VV^{1 - 2/p} \AA \right)^{-1} \aa_i \approx_{\alpha^{|1-2/p|}} \aa_i^T \left( \AA^T \WW^{1 - 2/p} \AA \right)^{-1} \aa_i.
\end{equation*}
Taking the $p/2$ power of both sides gives the desired result.
\end{proof}
An immediate corollary of this lemma is
\begin{corollary}
\label{cor:convergence}
Given $\AA$ with $\ell_p$ Lewis weights $\wwbar$ and a set of weights $\ww$ with $\ww \approx_{\alpha} \wwbar$,
\begin{equation*}
\textsc{LewisIterate}(\AA, p, \beta, \ww) \approx_{\beta^{p/2} \alpha^{|p/2 - 1|}} \wwbar.
\end{equation*}
\end{corollary}
\begin{proof}
This follows from applying Lemma~\ref{lem:contraction} to $\ww$ and $\wwbar$, noting that
$\textsc{LewisIterate}(\AA, p, \beta, \ww) \approx_{\beta^{p/2}} \textsc{LewisIterate}(\AA, p, 1, \ww)$ and that
$\textsc{LewisIterate}(\AA, p, 1, \wwbar) = \wwbar$.
\end{proof}

Interestingly, another corollary is that for $p < 4$, the Lewis weights exist and are unique:
\begin{corollary}
\label{cor:itlewis}
For all $p < 4$, there exists a unique assignment of weights $\ww$ that are a fixed point of
$\textsc{LewisIterate}(\AA, p, 1, _)$, or equivalently that satisfy Definition~\ref{def:lewis}.
\end{corollary}
\begin{proof}
This is just the Banach fixed point theorem applied to $\textsc{LewisIterate}(\AA, p, 1, _)$:
whenever $p < 4$, $|1-2/p| < 1$, so the iteration is a contraction mapping and has a unique
fixed point.
\end{proof}

We can also show that after one step $\ww$ is already polynomially close to $\wwbar$:
\begin{lemma}
\label{lem:initialization}
After $t = 1$ in $\textsc{ApproxLewisWeights}(\AA, p, \beta, T)$, $\ww_i \approx_{\beta^{p/2} n^{|p/2 - 1|}} \wwbar_i$.
\end{lemma}
\begin{proof}
We first note that we may assume, without loss of generality, that $\AA^T \WWbar^{1 - 2/p} \AA = \II$.
This is because both the definition of Lewis weights and the algorithms are invariant under replacing
$\AA$ with $\AA \RR$, with $\RR$ any full-rank $d$ by $d$ matrix.

Then we claim that
\begin{equation*}
\AA^T \AA \approx_{n^{|1 - 2/p|}} I.
\end{equation*}

First, we let $\bb_i = \wwbar_i^{1/2 - 1/p} \aa_i$, so that $\aa_i = \wwbar_i^{1/p - 1/2} \bb_i$.
We note that by the definition of leverage score and Lewis weight, $\norm{\bb_i}_2 = \ww_i$.

Now, consider any unit vector $\uu$.  We have
\begin{align*}
1 &= \uu^T \uu \\
&= \uu^T \BB^T \BB \uu \\
&= \sum_i (\uu_i^T \bb_i)^2 \\
&= \sum_i \wwbar_i (\wwbar_i^{-1} (\uu_i^T \bb_i)^2).
\end{align*}

On the other hand, we have
\begin{align*}
\uu^T \AA^T \AA \uu &= \sum_i \wwbar_i^{2/p - 1} (\uu_i^T \bb_i)^2 \\
&= \sum_i \wwbar_i^{2/p} (\wwbar_i^{-1} (\uu_i^T \bb_i)^2).
\end{align*}

Furthermore, $\sum_i \wwbar_i^{-1} (\uu_i^T \bb_i)^2 \leq n$, since each term is
$\leq 1$ (as, by Cauchy-Schwarz, $(\uu_i^T \bb_i)^2 \leq \norm{\bb_i}_2^2 = \wwbar_i$).

Then the worst-case distortion would be if $\sum_i \wwbar_i^{-1} (\uu_i^T \bb_i)^2 = n$ and all $\wwbar_i$ are $\frac{1}{n}$.  In
that case, the distortion is $n^{|1 - 2/p|}$, as desired.

Finally, the $\ww$ after one step are
\begin{equation*}
\ww_i \approx_{\beta^{p/2}} (\aa_i^T (\AA^T \AA)^{-1} \aa_i)^{p/2}.
\end{equation*}
These are then at most $\beta^{p/2} n^{p/2 |1-2/p|} = \beta^{p/2} n^{|p/2 - 1|}$ off from $\wwbar_i$.
\end{proof}

Combining this initial condition with the convergence result allows us to bound
the total number of steps, giving a proof of Theorem~\ref{thm:main}.

\begin{proof}[Proof of Lemma~\ref{lem:computeLewis}]
The total multiplicative contribution of the blowups from $\beta$ is at most $\beta^{\frac{p/2}{1 - |p/2 - 1|}}$.
The contribution from the starting error is at most $n^{|p/2 - 1|^T}$.

Then if $\beta \leq n^{\theta \frac{1 - |p/2 - 1|}{p}}$ and $T \geq \frac{\log(2 / \theta)}{1 - |p/2 - 1|}$,
each provides at most $n^{\theta / 2}$ error, so the result is an $n^{\theta}$-approximation.
\end{proof}

\section{Optimization Perspective on Lewis Weights}
\label{sec:convex}

Although simple and appealing, the iterative scheme described fails
when $p \geq 4$, since $\textsc{LewisIterate}$ stops being a contraction (further,
runtime approaches infinity as $p \to 4$).  To obtain Lewis weights for general
$p$, we therefore need another approach.  Here, we will give a characterization of
Lewis weights based on the solution to an optimization problem.

This optimization problem is essentially derived from a standard
proof of the existence of Lewis weights (\cite{BanachBook}, III.B, 7).  In fact,
it is interesting to note that both of these algorithms correspond fairly
directly to proofs of the existence and uniqueness of Lewis weights.
The optimization argument given here is valid for
all $p$, but only directly leads to an efficient algorithm for $p \geq 2$
(this is unimportant, since the iterative scheme is extremely efficient
when $p < 2$).

The optimization problem we will look at is, over symmetric matrices $\MM$,
\begin{equation*}
\begin{aligned}
& \underset{\MM}{\text{maximize}} & & \det M \\
& \text{subject to}
& & \sum_i (\aa_i^T \MM \aa_i)^{p/2} \leq d, \\
&&& \MM \succeq 0.
\end{aligned}
\end{equation*}

Note that when $p \geq 2$, the region in question is convex: it is the intersection
of the (convex) positive semidefinite cone and and $\ell_{p/2}$-norm constraint on
a vector of linear functions of $\MM$ (the $\aa_i^T \MM \aa_i$).  Maximizing the
determinant is also equivalent to minimizing the convex function $-\log \det \MM$,
so it is a convex optimization problem (we will give a specific algorithm using
convex optimization tools to find approximate Lewis weights below, in Theorem~\ref{thm:polyslow}.

\begin{lemma}
The optimization problem given above attains its maximum.
Further, for any matrix $\QQ$ reaching this maximum, the weights
\begin{equation*}
\ww_i = (\aa_i^T \QQ \aa_i)^{p/2}
\end{equation*}
satisfy the definition (Definition~\ref{def:lewis}) of Lewis weights.
\end{lemma}
\begin{proof}
First, note that the region in question is compact.  Since the function to be
optimized is continuous, it must attain a maximum.

Now, consider any matrix $\QQ$ that attains the maximum.
Then $\QQ$ will saturate the $\sum_i (a_i^T \MM a_i)^{p/2}$ constraint since otherwise
it could just be scaled up, but will be positive definite (not on the boundary of
the semidefinite cone) because otherwise it would have determinant 0.  We may then
say that $\QQ$ is a local maximum of the determinant subject to the constraint that
\begin{equation*}
\sum_i (\aa_i^T \MM \aa_i)^{p/2} = d.
\end{equation*}
Since all the functions in question are smooth, we may characterize such a local optimum
by Lagrange multipliers.  The gradient of the constraint (note that we are in the vector
space of symmetric matrices, so this is a matrix) at $\QQ$ can be seen to be
\begin{equation*}
p/2 \sum_i (\aa_i^T \QQ \aa_i)^{p/2 - 1} \aa_i \aa_i^T.
\end{equation*}
The gradient of the determinant, on the other hand, at $\QQ$, is $\det(\QQ) \QQ^{-1}$.
Lagrange multipliers then imply that at $\QQ$, $\QQ^{-1}$ must be parallel to
$\sum_i (\aa_i^T \QQ \aa_i)^{p/2 - 1} \aa_i \aa_i^T$.

The claim is that
\begin{equation*}
\ww_i = (\aa_i^T \QQ \aa_i)^{p/2}.
\end{equation*}
then must satisfy the conditions for Lewis weights.  To argue this, we note that
for these $\ww_i$, $\AA \WW^{1 - 2/p} \AA$ is
\begin{equation*}
\sum_i \ww_i^{1-2/p} \aa_i \aa_i^T = \sum_i (\aa_i^T \QQ \aa_i)^{p/2 - 1} \aa_i \aa_i^T.
\end{equation*}
In other words, $\AA \WW^{1 - 2/p} \AA$ is equal to $C \QQ^{-1}$ for some scalar $C$.
That implies that the $\ww_i$ could be defined as
\begin{equation*}
C^{p/2} (\aa_i^T (\AA \WW^{1 - 2/p} \AA)^{-1} \aa_i)^{p/2}.
\end{equation*}
Then some scaling of the $\ww_i$ satisfies
$\ww_i = (\aa_i^T (\AA \WW^{1 - 2/p} \AA)^{-1} \aa_i)^{p/2}$,
so $\ww$ is a multiple of Lewis weights. But since
\begin{equation*}
\sum_i \ww_i = \sum_i (\aa_i^T \MM \aa_i)^{p/2} = d
\end{equation*}
$\ww_i$ defined this way must be precisely the Lewis weights.
\end{proof}

\begin{corollary}
\label{cor:conlewis}
For all $p$, there exists an assignment of weights satisfying Definition~\ref{def:lewis}.
Furthermore, for $p \geq 2$, this assignment is unique.
\end{corollary}
\begin{proof}
The existence statement just follows from taking any of the $\QQ$ attaining the maximum.
For $p \geq 2$, the uniqueness follows from the fact that a strictly convex function
(which $-\log \det \MM$ is) attains a unique minimum on a convex set, and that applying
the argument in reverse every set of weights satisfying Definition~\ref{def:lewis}
is induced by such a $\QQ$.
\end{proof}
Of course, as mentioned earlier, uniqueness for $p < 2$ follows from Corollary~\ref{cor:itlewis}.

It is worth noting that this optimization problem gives a simple, geometric
characterization of the Lewis weights.  $\sum_i (\aa_i^T \MM \aa_i)^{p/2}$
is proportional to the average value of $\norm{\AA x}_p^p$ inside the ellipsoid
$x^T \MM^{-1} x \leq 1$.  Thus, the quadratic form induced by the Lewis weights
corresponds to the ellipsoid of maximum volume with a limited $p$-moment of
$\norm{\AA x}_p$; it is a ``softer'' version of a John ellipsoid, where the max
(effectively the $\infty$-moment) of $\norm{Ax}_p$ would be limited instead.

To use this algorithmically, we only assume an approximate solution.
We can analyze such a solution using the following lemma:
\begin{lemma}
There exists a constant $C$ such that for any $0 < \epsilon \leq 1$,
and positive semidefinite matrix $\MM$ satisfying
$\sum_i (\aa_i^T \MM \aa_i)^{p/2} = d$ and $\det \MM \geq (1-C \epsilon^2) \det \QQ$,
$\MM \approx_{1+\epsilon} \QQ$.
\end{lemma}
\begin{proof}
First note that $\QQ$ must be a maximum of
$\trace{\MM \QQ^{-1}}$ within the feasible region (otherwise, there would be a
direction in which the determinant increases around $\QQ$, making it not a local
maximum).  But for any matrix $\MM$ with $\trace{\MM \QQ^{-1}} \leq d$, if
$\MM \QQ^{-1}$ has any eigenvalues further than $\epsilon$ from 1,
$\det(\MM^{-1}) \leq (1 - O(\epsilon^2)) \det(\QQ^{-1})$.  Therefore,
an $O(\epsilon^2)$-approximate solution to the optimization problem implies
an $\epsilon$-approximation of $\QQ$ and of the Lewis weights themselves.
\end{proof}

The algorithmic guarantee is then:
\begin{theorem}
\label{thm:polyslow}
There exists a function $f(p, \epsilon)$ and a constant $C$ such that for any $\AA$, $p \geq 2$,
$(1+\epsilon)$-approximate Lewis weights (or the quadratic form $\QQ$) can be computed in time $O(f(p, \epsilon) n \log(n) \log \log(n) d^C)$.
\end{theorem}
Note that the $p < 2$ case already follows from Lemma~\ref{lem:computeLewis}.
\begin{proof}
For any positive real number $D$, consider the intersection of the sets $\MM \succeq 0$,
$\sum_i (\aa_i^T \MM \aa_i)^{p/2} \leq d$, and $\det(\MM) \geq D$.  This is a convex set
with a polynomial time separation oracle.

Furthermore, (if $D \leq \det \QQ$) this set always contains $\QQ$, which satisfies
$\QQ \preceq n^{1-2/p} \AA^T \AA$ by the argument from Lemma~\ref{lem:initialization}
(which nowhere assumes that $p < 4$) and is therefore contained in the ellipsoid
(over matrices!) $\norm{\QQ (\AA^T \AA)^{-1}}_F^2 \leq n^{2-4/p} d$.  If $D$ is within a
constant factor of $\det \QQ$, then the entire intersection is contained in a constant multiple
of that ellipsoid, and the volume ratio between it and the ellipsoid is at most $n^{O(d^2)}$.
Thus, for any such $D$, the ellipsoid algorithm can find an element in time $d^C \log n$ iterations,
each of which will take $n d^C$ time.

Finally, one may binary search (over an exponentially spaced set of possible $D$ values) to find
such a $D$ within an appropriate constant factor $1+O(\epsilon^2)$ in $\log \log n + \log d$ iterations.
\end{proof}
This is a polynomial-time algorithm, but not an input-sparsity one.


\section{Properties of Lewis weights}
\label{sec:properties}

Before describing our input-sparsity time algorithm arbitrary $\ell_p$,
we need to state some properties of Lewis weights.
Proofs of these properties are deferred to Appendix~\ref{sec:properties-proofs}.
In order to describe stability, we need to define a generalization of
the concept of Lewis weights: $\alpha$-almost Lewis weights

\begin{definition}
For a matrix $\AA$ and norm $\ell_p$, an assignment of weights $\ww$ is \emph{$\alpha$-almost Lewis} if
\begin{equation*}
\aa_i^T \left( \AA^T \WW^{1 - 2/p} \AA \right)^{-1} \aa_i \approx_{\alpha} \ww^{2/p}_i.
\end{equation*}
\end{definition}

The first property we investigate is \emph{stability}.  Recall that switching to $\ell_2$ via Lewis weights
gives a matrix $\BB = \WWbar^{1/2 - 1/p} \AA$, a reweighted version of $\AA$.  One may then ask how $\BB$ can change if the weights of the rows of $\AA$ change (as by multiplying by constants near 1).  We will give two equivalent definitions of stability, one in terms of $\alpha$-almost Lewis weights and the other in terms of $\BB$:

\begin{definition}
\label{def:stable}
For a value of $p > 0$, we will use the following two definitions of $c$-stability equivalently:
\begin{enumerate}
\item  For any $\AA'$ obtained by multiplying each row of $\AA$ by a number in
$[\alpha^{-1}, \alpha]$, the resulting reweighted $\BB'$ can be obtained by multiplying each row of
$\BB$ by a number in $[\alpha^{-c}, \alpha^c]$.
\label{def:stable1}
\item Any set of $\alpha$-almost Lewis weights $\ww$ for $\AA$ satisfy
\begin{equation*}
\wwbar_i \approx_{\alpha^c} \ww_i.
\end{equation*}
\label{def:stable2}
\end{enumerate}
\end{definition}

We start with stability results under multiplicative perturbations of rows.
Our results are in two regimes: a constant factor one for $p < 4$ and
a much weaker one for arbitrary $p > 2$.

\begin{lemma}
\label{lem:fullstable}
For all $p < 4$, Lewis weights are $\frac{p/2}{1 - |p/2 - 1|}$-stable.
\end{lemma}

\begin{lemma}
\label{lem:sqrtstable}
There exists a function $f(p)$ such that for all $p \geq 2$, $\ell_p$ Lewis weights are $O(f(p) \sqrt{d})$-stable.
\end{lemma}

Next, we give limits on the multiplicative factors by which that Lewis weights
can increase when additional rows are added to $\AA$.
Here, the best property applies for $p \leq 2$:

\begin{lemma}
\label{lem:fullmonotone}
For all $p \leq 2$, Lewis weights are \emph{monotonic}: adding an extra row to $\AA$
to yield $\AA'$ can never make the Lewis weights of existing rows go up.
\end{lemma}

Unfortunately, this fails to hold for $p > 2$.  However, the weights still can only increase by a bounded amount:
\begin{lemma}
\label{lem:weakmonotone}
For all $p > 2$, if $\AA'$ is $\AA$ with any number of extra rows added,
\begin{equation*}
\AA'^T \WWbar'^{1 - 2/p} \AA' \succeq d^{2/p - 1} \AA^T \WWbar^{1 - 2/p} \AA
\end{equation*}
and in particular no row in $\AA$ has its weight raised by more than $d^{p/2-1}$.
\end{lemma}


\section{An Input-sparsity Time Algorithm for General $\ell_p$}
\label{sec:sparseconvex}

This section gives an input-sparsity time algorithm for general $\ell_p$.
In particular, we claim that

\begin{theorem}
\label{thm:sparseconvex}
There exists an $f(p)$ and a constant $C$ such that given a matrix $\AA$, $p \in [1, \infty)$,
and $\theta < 1$, there is an algorithm that computes $n^\theta$-approximate $\ell_p$ Lewis
weights for $\AA$ in time $O(\frac{p}{\theta} \nnz(\AA) + f(p) d^{p/2 (1 + \theta) + C})$.
\end{theorem}

Just as with the previous algorithms, this can be combined with sampling by Lewis
weights to obtain an actual approximation for the matrix.

The core of the algorithm is a fairly simple recursion of the same style as
those in \cite{CohenLMMPS14:arxiv}.  For simplicity, we will write it in a form
that assumes the extremely weak condition that $\log n = O(\poly(d))$, although
even this mild constraint is unnecessary.  The recursive reduction ensures
that the algorithm from Theorem~\ref{thm:polyslow} only ever needs to run on
roughly $d^{p/2+1}$-sized samples.  We describe it as passing up a 2-approximation
of the inverse quadratic form $(\AA^T \WW^{1-2/p} \AA)^{-1}$.
Its pseudocode is given in Figure~\ref{fig:recurse}.

\begin{figure}[ht]

\begin{algbox}

$ \QQ = \textsc{ApproxLewisForm}(\AA, p, \theta)$

\begin{enumerate}
\item
If $n \leq d$, obtain $\QQ$ for $\AA$ (by Theorem~\ref{thm:polyslow}) and return $\QQ$ immediately.
\item
Uniformly sample $n / 2$ rows of $\AA$, producing $\widehat \AA$.
\item
Let $\widehat \QQ = \textsc{ApproxLewisForm}(\widehat \AA, p, \theta)$
\item
Set $\uu_i$ to $n^{\theta / p}$-approximate values of $\aa_i^T \widehat \QQ \aa_i$, computed using
the Johnson-Lindenstrauss lemma.
\item
Nonuniformly sample rows of $\AA$, taking expected $\pp_i = \min(1, f(p) n^{\theta / 2} d^{p/2} \log d \uu_i^{p/2})$
copies of row $i$ (each scaled down by $\pp_i^{-1/p}$), producing $\AA'$
\item
Obtain approximate inverse quadratic form $\QQ$ with Theorem~\ref{thm:polyslow}.
\item
Return $\QQ$
\end{enumerate}
\end{algbox}

\caption{Recursive procedure for approximating Lewis quadratic form}
\label{fig:recurse}

\end{figure}

The guarantees of this algorithm can be stated as
\begin{lemma}
\label{lem:approxlewisform}
With high probability, $\textsc{ApproxLewisForm}(\AA, p, \theta)$
\begin{enumerate}
\item returns $\QQ$ that's a 2-approximation of the true inverse quadratic
form $(\AA^T \WW^{1-2/p} \AA)^{-1}$, and
\label{part:givesapprox}
\item obtained $\QQ$ by invoking Theorem~\ref{thm:polyslow} on a matrix whose
expected row count is
\begin{equation*}
O\left(f(p) n^{\theta /2} d^{p/2 + 1} \log d\right).
\end{equation*}
\label{part:boundedrows}
\end{enumerate}
\end{lemma}

We begin by proving guarantees on the quadratic form proved.

\begin{proof}[Proof of Lemma~\ref{lem:approxlewisform}~Part~\ref{part:givesapprox}]
By Lemma~\ref{lem:weakmonotone}, the quadratic form of $\AA'$ is not bigger in any
direction by more than $d^{1-2/p}$ than that of $\AA$. Thus, the sampling probability
(when less than 1) of row $i$ is at least $f(p) d \log d \wwbar_i$.

Now, define a set of weights $\ww'$ in the sample as appropriate rescalings of the Lewis weights
of the original rows they are derived from (i.e. a row deriving from original row $i$ is assigned
weight $\frac{\wwbar_i}{\pp_i}$).  These rows then satisfy
\begin{equation*}
{\aa'}_j^T \left( \AA^T \WWbar^{1 - 2/p} \AA \right )^{-1} {\aa'}_j = {\ww'}_j^{2/p}.
\end{equation*}
Furthermore, by the ordinary ($\ell_2$) matrix Chernoff bound, with high probability
\begin{equation*}
{\AA'}^T {\WW'}^{1 - 2/p} {\AA'} \approx_{1+C / \sqrt{f(p) d}} \AA^T \WWbar^{1-2/p} \AA.
\end{equation*}
This implies that
\begin{equation*}
{\aa'_j}^T \left( {\AA'}^T {\WW'}^{1 - 2/p} {\AA'} \right )^{-1} \widehat \aa_j \approx_{1+C / \sqrt{f(p) d}} {\ww'}_j^{2/p}.
\end{equation*}
That is, the $\ww'$ are $(1+C / \sqrt{f(p) d})$-almost Lewis weights for $\AA'$.  Then by Lemma~\ref{lem:sqrtstable},
the true Lewis weights for $\AA'$, $\wwbar'$ are within a constant factor (of our choosing) of the $\ww'$
(with a correct setting of $f(p)$, which should be proportional to the \emph{square} of the $f(p)$
in Lemma~\ref{lem:sqrtstable}).  We finally see that ${\AA'}^T {\WWbar'}^{1-2/p} \AA'$ is within
a constant factor of ${\AA'}^T {\WW'}^{1-2/p} {\AA'}$ and thus within a constant factor
(which we can set to 2) of $\AA^T \WWbar^{1-2/p} \AA$.
\end{proof}

Next, we prove the bounds on the expected number of rows in the
intermediate matrix that we produce the final quadratic form.
This argument is similar to the proof of
Theorem 1 of~\cite{CohenLMMPS14:arxiv}.

\begin{proof}[Proof of Lemma~\ref{lem:approxlewisform}~Part~\ref{part:boundedrows}]

First, we consider a quantity similar to $\uu_i$: $\vv_i = \aa_i^T \widehat \QQ_i \aa_i$, where $\widehat \QQ_i$ is the
analogous Lewis quadratic form, but defined for $\widehat \AA_i$, $\widehat \AA$ with the row $\aa_i$ concatenated (if it was not already included).

Then $\vv_i^{p/2}$ is just the Lewis weight of $\aa_i$ in $\widehat \AA_i$.  Then analogous to \cite{CohenLMMPS14:arxiv}, we may argue
that the sum of the $\vv_i^{p/2}$ is the sum of the Lewis weights of the rows of $\widehat \AA$, plus $n/2$ times the average weight
of unincluded rows.  But the random process picking a subset of $n/2$ rows and then a random unincluded row is the same as picking
a random subset of $n/2+1$ rows and then a random ``special'' row from that subset.  We thus want to consider picking a random
subset of $n/2+1$ rows and then taking the Lewis weight of a random row within that subset.  But since the Lewis weights
for any matrix sum to at most $d$, the expected value of a random $\vv_i^{p/2}$ is at most $\frac{2d}{n+2}$, and the expected
sum of all the $\vv_i^{p/2}$ (including the rows included in $\widehat \AA$) is $\frac{2n+2}{n+2} d \leq 2d$.

So far, the argument has proceeded identically to \cite{CohenLMMPS14:arxiv}.  Now, however, we have to deal with the fact
that our weights are not defined based on $\vv_i^{p/2}$, but rather $\uu_i^{p/2}$.  The argument now is that
the Lewis weights for $\widehat \AA_i$, restricted to the $\widehat \AA$, are in fact $\frac{1}{1-\vv_i^{p/2}}$-almost
Lewis weights for $\widehat \AA$ (since the quadratic form can only shift by that factor).  Then if $\vv_i^{p/2}$ is smaller
than $\frac{1}{d^{p/2}}$ (in fact, even just for $\frac{1}{\sqrt{d}}$), Lemma~\ref{lem:sqrtstable} implies that
the true Lewis weights and true Lewis quadratic form are within $O(1)$ of what they are for $\widehat \AA$, and thus
that $\uu_i^{p/2}$ is at most a constant multiple of $\vv_i^{p/2}$, and $n^{\theta / 2} d^{p/2} \log d \uu_i^{p/2}$
is at most $O(n^{\theta / 2} d^{p/2} \log d) \vv_i^{p/2}$.  On the other hand, if $\vv_i^{p/2}$ is larger than
$\frac{1}{d^{p/2}}$, then 1 (which is also an upper bound on the sampling probability) is only
$d^{p/2}$ times $\vv_i^{p/2}$.  Thus, the expected sum of the actual sampling probabilities is at most
$O(n^{\theta / 2} d^{p/2} \log d)$ times larger than the expected sum of $\vv_i^{p/2}$, and is therefore
$O(n^{\theta /2} d^{p/2 + 1} \log d)$, up to a dependence on $p$.
\end{proof}

This, combined with a final Johnson-Lindenstrauss stage, satisfies the requirements of Theorem~\ref{thm:sparseconvex}.

It is worth noting that for $p < 2$, we may modify this analysis, replacing the use of Lemma~\ref{lem:weakmonotone} with
Lemma~\ref{lem:fullmonotone}; for $p < 4$, we may replace Lemma~\ref{lem:sqrtstable} with
Lemma~\ref{lem:fullstable}.  Using these, we can, for instance, for $p < 2$, run the algorithm with no $d^{p/2}$ factor
multiplying the weights, so that the samples $\widehat \AA$ are size $O(d \log d)$.  We may further use the iterative scheme
from Lemma~\ref{lem:computeLewis} (here, you can just compute leverage scores directly and exactly) instead of
Theorem~\ref{thm:polyslow}.  This can give an alternative input-sparsity time algorithm with a polynomial dependence $d^{\omega}$
up to polylog factors and tradeoffs with $\theta$, which can have better tradeoffs between the input-sparsity term
and polynomial term than directly running \ref{lem:computeLewis} with a fast leverage score approximation algorithm.


\section{$\ell_p$ Matrix Concentration Bounds}
\label{sec:ellp}

In this section we sketch proofs of the various matrix concentration
bounds that we utilize in our algorithms.

We get the core results from several earlier papers on $\ell_p$ approximation bounds.
To understand how these papers match our use of them, one should realize that they are
written in the language of approximate isomorphisms between Banach spaces.  Effectively,
they attempt to embed the column span of our matrix $\AA$ as a subspace of $L_p$.

\begin{theorem}
\label{thm:ellpsample}
Given an $n$ by $d$ matrix $\AA$
with $\ell_p$ Lewis weights $\wwbar$, for any set of sampling values $\pp_i$,
$\sum_i \pp_i = N$,
\begin{equation*}
\pp_i \geq f(d, N, p, \epsilon, \delta) \wwbar_i,
\end{equation*}
if we generate a matrix $\SS$ with $N$ rows, each chosen
independently as the $i$th standard basis vector,
times $\frac{1}{\pp_i^{1/p}}$, with probability $\frac{\pp_i}{N}$,
then with probability at least $1-\delta$ we have
\begin{equation*}
\norm{\SS \AA \xx}_1
\approx_{1 + \epsilon} \norm{\AA \xx}_1
\end{equation*}
for all vectors $\xx \in \Re^{d}$.

Valid asymptotic bounds on $d f(d, N, p, \epsilon, \delta)$ (i.e. the resulting row count,
since Lewis weights sum to $d$; here, the resulting row count itself
is plugged in as $N$) are given in the table below:

\begin{tabular}{| l | l | l |}
\hline
$p$ & $\delta$ & Sufficient row count \\ \hline
$p = 1$ & $\frac{1}{d^C}$ & $d \log(d / \epsilon) / \epsilon^2$ \\ \hline
$p = 1$ & $\frac{1}{C}$ & $d \log d / \epsilon^2$ \\ \hline
$1 < p < 2$ & $\frac{1}{C}$ & $d \log(d / \epsilon) \log(\log d / \epsilon)^2 / \epsilon^2$ \\ \hline
$p > 2$ & $\frac{1}{d^C}$ & $d^{p / 2} \log d \log(1 / \epsilon) / \epsilon^5$ \\ \hline
\end{tabular}
\end{theorem}

The last entry was proved directly in~\cite{BourgainLM89}.  For the others (the $p < 2$ cases),
these statements (that this sampling procedure is valid) are not proved directly.  Rather,
they look at the case where the Lewis weights are \emph{uniformly} small,
examining a random process choosing $\sigma_i$ independent Rademacher variables
(i.e. independently $\pm 1$ with probability $\frac{1}{2}$ each) to give
\begin{equation*}
\max_{\norm{\AA \xx}_p = 1} \left | \sum_i \sigma_i |\aa_i^T \xx_i|^p \right |.
\end{equation*}
It is worth noting that this is equivalent to taking the error of of sampling
about half the rows by unbiased coin flips.

We may, specifically, get the following statements from the existing results.

The following was shown in~\cite{Talagrand90}
\begin{lemma}[\cite{Talagrand90}]
\label{lem:talagrandL1}
There exists a constant $C$ such that if every row of $\AA$ has
Lewis weight at most $C \frac{\epsilon^2}{\log d}$, we have
\begin{equation*}
\expec{\sigma}{\max_{\norm{\AA \xx}_1 = 1}
\left | \sum_i \sigma_i \left|\aa_i^T \xx \right| \right|} \leq \epsilon.
\end{equation*}
\end{lemma}

\begin{proof}
This is implicit in \cite{Talagrand90}.  First, the proof of Proposition 1 in that paper includes a proof that this quantity is dominated by a constant times
\begin{equation*}
\expec{g}{\max_{\norm{\AA \xx}_1 = 1}
\left | \sum_i g_i ( \aa_i^T \xx ) \right |},
\end{equation*}
where the $g_i$ are independent standard Gaussian variables.  Note that this differs in two ways: the Rademacher variables are replaced with Gaussians, and
the inner absolute values have been removed (so that it is now just the max of a linear function). 
The first step is done via the comparison lemma for Radamacher processes,
which we state formally at the start of Section~\ref{sec:concentration}.
The second step is by the contraction principle, which in turn relies on the
convexity of the quantity being bounded.

This latter quantity is then bounded, under this assumption of bounded Lewis weights,
in the proof of Proposition 2.
The probability measure $\nu$, as defined on the bottom of page 366
corresponds to the probability distribution of rows proportional to
their Lewis weights.
It is then split into atoms with small Lewis weights, and shows that randomly
sampling a subset of these atoms gives the bound above.
Our setup in Lemma~\ref{lem:talagrandL1} is equivalent to this
situation after splitting.
This means our result follows from the same proof as on page 367.
This leads to a result as stated in Proposition $2$ with $n = d$,
$K(X) = O(\sqrt{\log{d}})$
and $C \frac{\epsilon^2}{\log d}$.
\end{proof}

The following was shown in~\cite{TalagrandLp}
\begin{lemma}[\cite{TalagrandLp}]
For any $p < 2$, there exists a constant $C$ such that if every row of $\AA$ has
Lewis weight at most
$C \frac{\epsilon^2}{\log (n) (\log (\log n / \epsilon))^2}$, we have
\begin{equation*}
\expec{\sigma}{\max_{\norm{\AA \xx}_p = 1}
\left | \sum_i \sigma_i \left|\aa_i^T \xx \right|^p \right|} \leq \epsilon.
\end{equation*}
\end{lemma}

\begin{proof}
This was proven in Proposition 2.3 in~\cite{TalagrandLp}.
It gives a bound of the form
\begin{equation*}
\Lambda_{F} \leq K \sqrt{ \frac{n}{M} \log{M} } \left( \log\log{M} + \log \left( \frac{M}{n} \right) \right).
\end{equation*}
where $\Lambda{F}$ was defined in Proposition 2.2 as
\begin{equation*}
\Lambda_{F} 
= \expct{\sigma} {\max_{\xx \in F_1, \norm{\xx}_p \leq 1}
\left | \sum_i \lambda_i \sigma_i \left| \xx(i) \right|^p \right|}.
\end{equation*}
Here $\lambda_i$ is $\wwbar_i / d$ where $\wwbar_i$ are the Lewis weights as defined in this paper, and $x(i)$
has been normalized (divided by $\lambda_i^{1/p}$).

Talagrand's $n$ is our $d$ and Talagrand's $M$ is our $n$.
However, Talagrand's proof is only using the $\frac{n}{M}$ as an upper bound on the Lewis weights;
thus, we can get our claim.  With a Lewis weight (by our definition) upper bound of $U$, expressions such as
$\sqrt{\frac{n \log M}{M}}$ can be replaced with $\sqrt{U \log M}$, and $\frac{1}{M \log^4 M}$ with
$\frac{U}{n \log^4 M}$.
\end{proof}
\cite{TalagrandLp} conjectures, immediately after its core claim Theorem 1.1,
that the extra $\log \log n$ factors are unnecessary for this result
(it describes them as ``truly parasitic'').  It also points out the difficulty of adapting its approach to remove them.

We will give a general (for $p \leq 2$) reduction, Lemma~\ref{lem:momentreduct}, of these moment bounds to moment bounds on the error of our described sampling procedure.


Now we may give the general reduction.  We describe a version with an extra
restriction that the required Lewis weight bound is larger than $O(1/d)$;
this is not actually necessary (and we will sketch how to avoid this).
\begin{restatable}{lemma}{momentreduct}
\label{lem:momentreduct}
There exist constants $C_1$, $C_2$ such that for any $p < 2$, if a uniform bound of $\frac{1}{g(p, n, d, \epsilon, \delta)} = \Omega(1/d)$ on
Lewis weights implies
\begin{equation*}
\expec{\sigma}{\left (\max_{\norm{\AA \xx}_p = 1} \left | \sum_i \sigma_i |\aa_i^T \xx_i|^p \right | \right )^l} \leq \epsilon^l \delta,
\end{equation*}
then sampling as described in Theorem~\ref{thm:ellpsample} with $\pp_i \geq g(p, N + C_1 d^2, d, \epsilon / C_2, \delta) \wwbar_i$ will satisfy
\begin{equation*}
\expec{\SS}{\left ( \max_{\norm{\AA \xx}_p = 1} | \norm{\SS \AA \xx}_p^p - 1 | \right )^l} \leq \epsilon^l \delta
\end{equation*}
and in particular, $\norm{\SS \AA \xx}_p \approx_{1 + \epsilon} \norm{\AA \xx}_p$ with probability at least $1-\delta$.
\end{restatable}
This is proved in Appendix~\ref{sec:reduct}.

The following was shown in~\cite{BourgainLM89}
\begin{lemma}[\cite{BourgainLM89}]
For any $p > 2$, there exists a constant $C$ such that if
if we sample $O(d^{p / 2} \log d \log(1 / \epsilon) / \epsilon^5)$ rows
with probability proportional to Lewis weights to give $\AA'$, we have
\begin{equation*}
\max_{\norm{\AA \xx}_p = 1}
\left | \norm{\AA' \xx}_p^p - 1 \right |  \leq \epsilon
\end{equation*}
with high probability.
\end{lemma}

\begin{proof}
This is implicitly shown in the proof of Theorem 7.3 in~\cite{BourgainLM89}.  In particular, line
7.27 is giving sufficient conditions for a Bernstein inequality, applied to random samples,
to imply approximation.  These samples are chosen according to the probability measure $\nu$
used in the paper, which is proportional to the Lewis weights.  We are only assuming the weights
are lower bounds for the sampling probabilities, but increasing the probabilities further can
only improve the condition of 7.27.
\end{proof}

There are several additional proofs of results similar to those we refer to here, including in Chapter 15 of
\cite{LedouxTBook} and in \cite{GuedonR07}.


\section{An Elementary Proof of $\ell_1$ Matrix Concentration Bound}
\label{sec:concentration}

In this section, we give an elementary proof of the concentration bound
described in Theorem~\ref{thm:l1chernoff}.

\lonechernoff*

By Lemma~\ref{lem:momentreduct}, it suffices to prove that there exists a $C_r$ such that
if \emph{every} Lewis weight is at most $C_r \epsilon^2 / \log(n / \delta)$, there exists an $l$ such that
\begin{equation*}
\expec{\sigma}{\left (\max_{\norm{\AA \xx}_1 = 1} \left | \sum_i \sigma_i |\aa_i^T \xx_i| \right | \right )^l} \leq \epsilon^l \delta
\end{equation*}
where the $\sigma_i$ are independent Rademacher variables.

We begin by invoking the comparison theorem for Radamacher processes,
as stated in Proposition 1 of~\cite{LedouxT89}.

\begin{lemma}
\label{lem:comparison}
For any positive, monotonic function $f$ we have:
\begin{equation*}
\expct{\sigma}{ f \left( \max_{\xx, \norm{\AA \xx}_1 = 1 } \abs{\sum_i \sigma_i \abs{\aa_{i}^T \xx }} \right)}
\leq 2 \expct{\sigma}{  f \left( \max_{\xx, \norm{\AA \xx}_1 = 1 } \sum_i \sigma_i \aa_{i}^T \xx \right) }.
\end{equation*}
\end{lemma}

It remains to bound this new random process (with the absolute values removed).
Here note that duality of norms gives that
\begin{equation*}
\max_{\norm{\yy}_{1} \leq 1} \xx^T \yy
= \norm{\xx}_{\infty}.
\end{equation*}

\begin{lemma}
\label{lem:lewlinf}

\begin{equation*}
\left( \max_{\xx, \norm{\AA \xx}_1 = 1} \sum_i \sigma_i \aa_i^T \xx \right)^{l}
\leq \sum_i \left | \sum_j \sigma_j
\wwbar_i^{-1} \aa_i^T ( \AA^T \WWbar^{-1} \AA )^{-1} \aa_j \right |^l .
\end{equation*}

\end{lemma}

\begin{proof}

Consider the expression
\begin{equation*}
\max_{\xx, \norm{\AA \xx}_1 = 1 } \sum_i \sigma_i \aa_i^T \xx.
\end{equation*}
Direct algebraic manipulation gives that it is equivalent to
\begin{equation*}
= \max_{\xx, \norm{ \AA \xx}_1 \leq 1} \ssigma^T \AA \xx
\end{equation*}
with $\ssigma$ representing a vector of all the $\sigma_i$.

Now, let $\PPi = \AA ( \AA^T \WWbar^{-1} \AA )^{-1} \AA^T \WWbar^{-1}$.  $\PPi$ is a projection matrix to
the column space of $\AA$, and in particular $\PPi \AA = \AA$ (this can be shown simply by
multiplying it out and cancelling $\AA^T \WWbar^{-1} \AA$ with its inverse).  Note that $\PPi$ is
specifically the \emph{orthogonal} projection to this subspace with respect to the
$\wwbar^{-1}$-weighted inner product: the natural inner product induced by the Lewis weights.

Thus we have $\PPi \AA \xx = \AA \xx$, so the quantity is equal to
\begin{equation*}
\max_{\xx, \norm{ \AA \xx}_1 \leq 1} \ssigma^T \PPi \AA \xx.
\end{equation*}
This is upper bounded by
\begin{equation*}
\max_{\yy, \norm{\yy}_1 \leq 1} \ssigma^T \PPi \yy.
\end{equation*}
as any $\AA \xx$ can be plugged in as $\yy$.  But by duality of norms, as mentioned above,
that quantity is at most
\begin{equation*}
\norm{\PPi^T \ssigma}_\infty = \max_i \left | \sum_j \sigma_j \wwbar_i^{-1} \aa_i^T ( \AA^T \WWbar^{-1} \AA )^{-1} \aa_j \right |.
\end{equation*}

The $l$\textsuperscript{th} power of this max is at most the sum of the
$l$\textsuperscript{th} powers of the entries.
This gives
\begin{equation*}
\left( \max_i \left | \sum_j \sigma_j \wwbar_i^{-1} \aa_i^T ( \AA^T \WWbar^{-1} \AA )^{-1} \aa_j \right | \right)^l
\leq \sum_i \left | \sum_j \sigma_j \wwbar_i^{-1} \aa_i^T ( \AA^T \WWbar^{-1} \AA )^{-1} \aa_j \right |^{l}.
\end{equation*}
\end{proof}

To bound these terms, we will use the Khintchine inequality and the definition of Lewis weights.

\begin{lemma}[Khintchine]
Let $\sigma_i$ be independent Radamacher random variables.
Let $1 \leq l < \infty$  and let $\xx_i \in \Re$. Then there exists an absolute constant $C$ such that
\begin{equation*}
\expct{\sigma}{ \left | \sum_{i} \sigma_i \xx_i \right |^{l}}
\leq \left(Cl \sum_{i} \left| x_i \right|^2 \right)^{l/2}.
\end{equation*}
\end{lemma}

\begin{lemma}
\label{lem:momBound}
There exists an absolute constant $C$ such that, for any $\AA$ having all Lewis weights bounded by $U$, and $l \geq 1$
\begin{equation*}
\expec{\sigma}{\left (\max_{\norm{\AA \xx}_1 = 1} \left | \sum_i \sigma_i |\aa_i^T \xx_i| \right | \right )^l} \leq 2n (Clu)^{l/2}.
\end{equation*}
\end{lemma}

\begin{proof}
We apply Khintchine's inequality to the terms
$\expec{\sigma}{\left | \sum_j \sigma_j \wwbar_i^{-1} \aa_i^T ( \AA^T \WWbar^{-1} \AA )^{-1} \aa_j \right |^{l}}$,
which gives an upper bound for each term of
\begin{equation*}
\left ( Cl \sum_j ( \wwbar_i^{-1} \aa_i^T ( \AA^T \WWbar^{-1} \AA )^{-1} \aa_j )^2 \right )^{l/2}.
\end{equation*}

Now, we apply our Lewis weight upper bound to show that
\begin{equation*}
\sum_j ( \wwbar_i^{-1} \aa_i^T ( \AA^T \WWbar^{-1} \AA )^{-1} \aa_j )^2 \leq U \sum_j ( \wwbar_i^{-1} \aa_i^T ( \AA^T \WWbar^{-1} \AA )^{-1} \aa_j \wwbar_j^{-1/2} )^2.
\end{equation*}

We may then rewrite that expression as
\begin{equation*}
U \wwbar_i^{-2} \sum_j  \aa_i^T ( \AA^T \WWbar^{-1} \AA )^{-1} \aa_j \wwbar_j^{-1} \aa_j^T ( \AA^T \WWbar^{-1} \AA )^{-1} \aa_i.
\end{equation*}

Only the middle, $\aa_j \wwbar_j^{-1} \aa_j$ depends on $j$, so the whole sum is
\begin{equation*}
U \wwbar_i^{-2} \aa_i^T ( \AA^T \WWbar^{-1} \AA )^{-1} \left ( \sum_j \aa_j \wwbar_j^{-1} \aa_j^T \right ) ( \AA^T \WWbar^{-1} \AA )^{-1} \aa_i.
\end{equation*}
Since $\sum_j \aa_j \wwbar_j^{-1} \aa_j^T = \AA^T \WWbar^{-1} \AA$, this is just
\begin{equation*}
U \wwbar^{-2} \aa_i^T ( \AA^T \WWbar^{-1} \AA )^{-1} \aa_i = U \wwbar^{-2} \wwbar^2
\end{equation*}
or just $U$.  Note that the last equation used the definition of Lewis weights.  Thus, for all $i$,
\begin{equation*}
\expec{\sigma}{\left( \sum_j \sigma_j \wwbar_i^{-1} \aa_i^T ( \AA^T \WWbar^{-1} \AA )^{-1} \aa_j \right )^{l}}
\leq (ClU)^{l/2},
\end{equation*}
so their sum
\begin{equation*}
\expec{\sigma}{\sum_i \left( \sum_j \sigma_j \wwbar_i^{-1} \aa_i^T ( \AA^T \WWbar^{-1} \AA )^{-1} \aa_j \right )^{l}} \leq n (ClU)^{l/2}.
\end{equation*}
Combining this with Lemma~\ref{lem:comparison} and Lemma~\ref{lem:lewlinf} gives the desired bound.
\end{proof}

We can now prove the main theorem.

\begin{proof}[Proof of Theorem~\ref{thm:l1chernoff}]
First, we show that there exists an absolute constant $C$ such that with any $\AA$ having each Lewis weight
bounded by $\frac{C_r \epsilon^2}{\log(2n/\delta)}$, there exists an $l$ such that,
\begin{equation*}
\expec{\sigma}{\left (\max_{\norm{\AA \xx}_1 = 1} \left | \sum_i \sigma_i |\aa_i^T \xx_i| \right | \right )^l} \leq \epsilon^l \delta.
\end{equation*}

To obtain this, we simply plug $l = \log(2n/\delta)$ and $U = \frac{1}{Ce^2} \epsilon^2 / \log(2n / \delta)$ into
Lemma~\ref{lem:momBound}, giving $(ClU)^{1/2} = \frac{\epsilon}{e}$,
and $(ClU)^{l/2} = \epsilon^l \frac{\delta}{2n}$.

Plugging this into Lemma~\ref{lem:momentreduct} and applying Markov's inequality gives Theorem~\ref{thm:l1chernoff}.
\end{proof}

We remark that Talagrand used a slightly different method to go from bounds for this Rademacher process to
the existence of good subspace approximations.  Essentially, he pointed out that this process shows that
$\AA$ with large $n$ are well-approximated by ones with about $\frac{3}{4} n$ rows, with the approximation
quality improving as $n$ gets larger.  Iteratively replacing $\AA$ with a smaller approximation eventually leaves
one of size $O(d \log d / \epsilon^2)$ (in Talagrand's case; with our approach we would again get $d \log(d/\epsilon) / \epsilon^2$).

\section*{Acknowledgements}
We thank Jon Kelner and James Lee for guiding us through prior
works on this subject, and advice during the writing of this manuscript.
We also acknowledge Yin-Tat Lee, Jelani Nelson, Ilya Razenshteyn, Aaron Sidford and David Woodruff
for helpful discussions.

\bibliographystyle{alpha}
\bibliography{ref}

\begin{appendix}


\section{Properties of Lewis weights - proofs}
\label{sec:properties-proofs}
This appendix contains proofs of the properties stated in Section~\ref{sec:properties}.

\begin{lemma}
\label{lem:equivstable}
Parts \ref{def:stable1} and \ref{def:stable2} of Definition~\ref{def:stable} are equivalent.
\end{lemma}
\begin{proof}
Given $\alpha$-almost Lewis weights $\ww$, define the errors $\ee$ as
\begin{equation*}
\ee_i = \aa_i^T \left( \AA^T \WW^{1 - 2/p} \AA \right)^{-1} \aa_i \ww_i^{-2/p}.
\end{equation*}
Then define $\AA'$ as $\EE^{1/p-1/2} \AA$ and $\BB'$ as $\WW^{1/2-1/p} \AA$. $\AA'$ has
Lewis weights $\ee_i \ww_i$, which reweight it to $\BB'$.  $\AA$,
when reweighted by the true weights $\wwbar$, will give $\BB = \WWbar^{1-2/p} \AA$.
$\AA$ is within a factor of $\alpha^{|1/2-1/p|}$ of $\AA'$; part \ref{def:stable1} then would imply
that $\BB$ is within a $\alpha^{c |1/2-1/p|}$ factor of $\BB$ and equivalently that
\begin{equation*}
\wwbar_i \approx_{\alpha^c} \ww_i
\end{equation*}
giving \ref{def:stable2}.
This argument works in reverse as well (interpreting the Lewis weights for any $\AA'$ as scaled
approximate Lewis weights for $\AA$), so $c$-stability is equivalent to the statement that
$\alpha$-almost Lewis weights are always within $\alpha^c$ of the true Lewis weights.  Note
that this implies that $(1+O(1/c))$-almost Lewis weights are within $O(1)$ of the true Lewis
weights.
\end{proof}

Now, we prove all the results from Section~\ref{sec:properties}.
\begin{proof}[Proof of Lemma~\ref{lem:fullstable}]
Here, we prove the $\alpha$-almost weight notion of stability, part~\ref{def:stable1} of Definition~\ref{def:stable}.
First, note that for $\alpha$-almost Lewis weights $\ww$,
\begin{equation*}
\textsc{LewisIterate}(\AA, p, 1, \ww) \approx_{\alpha^{p/2}} \ww.
\end{equation*}

In other words, the first step can only change the weights by a power of $p/2$.
But then since $\textsc{LewisIterate}$ is a contraction mapping by $|p/2 - 1|$,
the $t$\textsuperscript{th} step after this, when iterating, can change by only $\alpha^{p/2 |p/2 - 1|^t}$.
The iteration will converge to $\wwbar$, while changing by a total of at most
\begin{equation*}
\alpha^{p/2 \sum_{t=0}^\infty |p/2 - 1|^t} = \alpha^{\frac{p/2}{1 - |p/2 - 1|}}.
\end{equation*}
\end{proof}

\begin{proof}[Proof of Lemma~\ref{lem:sqrtstable}]
Here, we will find it more convenient to use the definition of stability in terms of a reweighted matrix,
part~\ref{def:stable2} of Definition~\ref{def:stable}.

We bound the stability by bounding it infinitesimally.  We let $\rr_i$ and $\ss_i$ be factors that can be
multiplied by the weights of the row, constrained so that $\SS \AA$, when multiplied by
$\WWbar^{1/2-1/p}$ for its Lewis weights $\WWbar$, gives $\RR \AA$.  We will bound the derivative
of $\rr$ with respect to $\ss$.  For any direction vector $\Delta(\ss)$, let the derivative
in that direction by $\Delta(\rr)$.  Then the stability claim is that the max of $\frac{\Delta(\ss)_i}{\ss_i}$
is at most $c$ times the max of $\frac{\Delta(\rr)_i}{\rr_i}$.

It is not apparent how to bound this directly.  However, note that saying that the derivative of $\rr$ with
respect to $\ss$ in the direction of $\Delta(\ss)$ equals $\Delta(\rr)$ is equivalent to saying that the derivative of $\ss$ with
respect to $\rr$ in the direction of $\Delta(\rr)$ equals $\Delta(\ss)$.  Thus, what we want is to show that
that the ratio can't \emph{decrease} too much when differentiating with respect to $\rr$.

For convenience, we will assume that $\AA^T \RR^2 \AA = \II$ (by a change of basis)
and $\norm{\aa_i}^2 = 1$ (by rescaling the rows of $\AA$ simultaneously with $\RR$).

The $\ss$ in terms of the $\rr$ are
\begin{equation*}
\ss_i = \rr_i^{2/p} (\aa_i^T (\AA^T \RR^2 \AA)^{-1} \aa_i)^{1/p - 1/2}.
\end{equation*}

The initial $\ss_i$ are just $\rr_i^{2/p}$.

Then given that $\norm{\aa_i}^2 = 1$ and $\AA^T \AA = \II$, the partial derivative of $\ss_i$ with respect to $\rr_j$ is
\begin{equation*}
\rr_i^{2/p-1} ( (1-2/p) \rr_i \rr_j (\aa_i^T \aa_j)^2 + 2/p \delta_{i,j} ).
\end{equation*}

We may define the symmetric matrix $\MM_{i,j} = ( (1-2/p) \rr_i \rr_j (\aa_i^T \aa_j)^2 + 2/p \delta_{i,j} )$.  Then the partial derivatives gives us
\begin{equation*}
\Delta(\ss) = \RR^{2/p-1} \MM \Delta(\rr).
\end{equation*}

The matrix of $\rr_i \rr_j \rr_j (\aa_i^T \aa_j)^2$ is positive semidefinite--we can write it as $\CC^T \CC$,
where the rows of $\CC$ are scaled tensor squared rows of $\AA$: $\cc_i = \rr_i (\aa_i \otimes \aa_i)$.
Thus, $\MM \succeq 2/p \II$, or equivalently, all eigenvalues of $\MM$ are at least $\frac{2}{p}$.
Thus,
\begin{equation*}
\norm{\MM \Delta(\rr)}_2 \geq \frac{2}{p} \norm{\Delta(\rr)}_2.
\end{equation*}

In other words,
\begin{equation*}
\norm{\RR^{1-2/p} \Delta(\ss)}_2 \geq \frac{2}{p} \norm{\Delta(\rr)}_2.
\end{equation*}

The elements of $\RR^{1-2/p} \Delta(\ss)$ are $\rr_i \frac{\Delta(\ss)_i}{\rr_i^{2/p}}$
$\rr_i \frac{\Delta(\ss)_i}{\ss_i}$.  Finally, since $\sum_i \rr_i^2 = d$,
\begin{equation*}
\norm{\RR^{1-2/p} \Delta(\ss)}_2 \leq \sqrt{d} \max_i \frac{\Delta(\ss)_i}{\ss_i},
\end{equation*}
so that
\begin{equation}
\label{eqn:bigl2}
\max_i \frac{\Delta(\ss)_i}{\ss_i} \geq \frac{2}{p \sqrt{d}} \norm{\Delta(\rr)}_2.
\end{equation}

However, we still need to address the possibility that $\norm{\Delta(\rr)}_2$ is much smaller
than $\max_i \frac{\Delta(\rr)_i}{\rr_i}$.  To do this, consider the $i$ that maximizes that
ratio (and assume by scaling that the ratio is 1).  Then, for that particular $i$,
\begin{equation*}
\frac{\Delta(\ss)_i}{\ss_i} = \frac{2}{p} - (1-2/p) \sum_j \rr_j^2 (\aa_i^T \aa_j)^2 \left ( \frac{\Delta(\rr)_j}{\rr_j} \right ).
\end{equation*}
Note that we can also express $\norm{\Delta(\rr)}_2^2$ as
\begin{equation*}
\sum_j \rr_j^2 \left ( \frac{\Delta(\rr)_j}{\rr_j} \right )^2.
\end{equation*}
Now, $0 \leq \rr_j^2 (\aa_i^T \aa_j)^2 \leq \rr_j^2$, while $\sum_j \rr_j^2 (\aa_i^T \aa_j)^2 = 1$
(since $\AA^T \RR^2 \AA = \II$), so
$\sum_j \rr_j^2 (\aa_i^T \aa_j)^2 \left ( \frac{\Delta(\rr)_j}{\rr_j} \right )$ is at most the average
of $\left ( \frac{\Delta(\rr)_j}{\rr_j} \right )$ in the largest mass-1 set of $j$s.  Thus
it is at most $\norm{\Delta(\rr)}_2$ (since weighted $\ell_2$ norms dominate $\ell_1$ norms when the total
mass sums to 1).
Plugging that in we get
\begin{equation}
\label{eqn:smalll2}
\max_i \frac{\Delta(\ss)_i}{\ss_i} \geq 2/p - (1-2/p) \norm{\Delta(\rr)}_2.
\end{equation}

With $\max_i \frac{\Delta(\rr)_i}{\rr_i} = 1$, we have both
Equation~\ref{eqn:bigl2} and Equation~\ref{eqn:smalll2}.
The former is increasing in $\norm{\Delta(\rr)}_2$ and the latter decreasing,
so the worst case for the minimum of the two of these is when they are equal,
which occurs at $\norm{\Delta(\rr)}_2 = \frac{\frac{2}{p}}{\frac{2}{p \sqrt{d}} + (1-2/p)}$.
This bound is
\begin{align*}
\max_i \frac{\Delta(\ss)_i}{\ss_i} &\geq \frac{\frac{4}{p^2 \sqrt{d}}}{\frac{2}{p \sqrt{d}} + (1-2/p)} \\
&= \frac{4}{2p + (p^2-2p) \sqrt{d}} \\
&= \Omega \left ( \frac{1}{\sqrt{d}} \right ).
\end{align*}
\end{proof}

\begin{proof}[Proof of Lemma~\ref{lem:fullmonotone}]
Assume $p < 2$ (the case for $p=2$ follows from the fact
that $\ell_2$ Lewis weights are just leverage scores).
Let $r$ be the maximum ratio (over $i$) of $\wwbar'_i$ to $\wwbar_i$, which occurs at some row $i^*$.
We want to show than $r \leq 1$; we will suppose that $r > 1$ and obtain a contradiction.

Now, we have
\begin{equation*}
\AA'^T \WWbar'^{1-2/p} \AA' \succeq r^{1-2/p} \AA^T \WWbar^{1-p/2} \AA,
\end{equation*}
since each row of $\AA$ contributes at least $r^{1-2/p}$ as much to the left hand side as to the right.  But the
$\wwbar'_i$ are equal to $(\aa_i^T (\AA'^T \WWbar'^{1-2/p} \AA')^{-1} \aa_i)^{p/2}$, and
\begin{align*}
(\aa_i^T (\AA'^T \WWbar'^{1-2/p} \AA')^{-1} \aa_i)^{p/2} &\leq r^{-p/2 (1-2/p)} (\aa_i^T (\AA^T \WWbar^{1-2/p} \AA)^{-1} \aa_i)^{p/2} \\
&= r^{1-p/2} \wwbar_i
\end{align*}
Thus, in particular, no weight increases by more than $r^{1-p/2}$.  But with $p < 2$, if $r > 1$,
$r^{1-p/2} < r$.  Thus,
\begin{equation*}
\wwbar'_{i^*} \leq r^{1-p/2} \wwbar_{i^*} < r \wwbar_{i^*}
\end{equation*}
This contradicts the fact that $r = \frac{\wwbar'_{i^*}}{\wwbar_{i^*}}$.
\end{proof}

\begin{proof}[Proof of Lemma~\ref{lem:weakmonotone}]
First, assume without loss of generality that $\AA^T \WWbar^{1 - 2/p} \AA$ is equal to the identity matrix.
Then the claim is that all eigenvalues of $\AA'^T \WWbar'^{1 - 2/p} \AA'$ are at least $d^{2/p - 1}$.

For convenience, define $\PP = \AA'^T \WWbar'^{1 - 2/p} \AA'$, and let $\uu$ be an eigenvector of
$\PP$ of minimum eigenvalue, normalized so that $\norm{u}_2 = 1$.  Let $\lambda$ be the eigenvalue
of $\uu$, $\uu^T \PP \uu$.  Then one may see by looking at an orthogonal eigenbasis that
\begin{equation*}
\PP^{-1} \succeq \lambda^{-1} \uu \uu^T.
\end{equation*}

We further define normalized rows $\vv_i = \ww_i^{-1/p} \aa_i$, so that $\AA^T \WWbar^{1 - 2/p} \AA$,
the identity matrix, may be written as
\begin{equation*}
\sum_i \wwbar_i \vv_i \vv_i^T
\end{equation*}
with the $\vv_i$ each being unit vectors.  $\PP$ is lower bounded by its contributions from the original rows
($\AA^T \WWbar'^{1-2/p} \AA$); expressing the latter in terms of the $\vv$ gives
\begin{align*}
\PP &\succeq \sum_i \wwbar_i \frac{(\wwbar'_i)^{1-2/p}}{\wwbar_i^{1-2/p}} \vv_i \vv_i^T \\
&= \sum_i \wwbar_i (\vv^T \PP^{-1} \vv)^{p/2-1} \vv_i \vv_i^T.
\end{align*}

Since $\PP^{-1} \succeq \lambda^{-1} \uu \uu^T$ this further gives
\begin{equation*}
\PP \succeq \sum_i \ww_i \lambda^{1-p/2} (\uu^T \vv)^{p-2} \vv_i \vv_i^T.
\end{equation*}
What we have done is provide a lower bound for the quadratic form $\PP$ given that it has
a small eigenvalue, using the fact that relative leverage scores cannot have decreased
below the leverage scores of the restriction to the eigenspace.  This is not affected by
any increases in the other eigenvalues,

Plugging $\uu$ into the quadratic form implies
\begin{equation*}
\uu^T \PP \uu \geq \lambda^{1-p/2} \wwbar_i (\uu^T \vv)^p.
\end{equation*}

Now, the original quadratic form was the identity, so $\sum_i \ww_i = d$ and $\sum_i \wwbar_i (\uu^T \vv)^2 = 1$.
By a weighted power-mean inequality, $\sum_i \wwbar_i (\uu^T \vv)^p \geq d^{1-p/2}$.

Plugging in $\lambda = \uu^T \PP \uu$ finally gives
\begin{align*}
\lambda &\geq \lambda^{1-p/2} d^{1-p/2} \\
\lambda^{p/2} &\geq d^{1-p/2} \\
\lambda &\geq d^{2/p - 1}.
\end{align*}
\end{proof}

\section{Reduction to Uniform Weight Sampling}
\label{sec:reduct}
Here, we prove Lemma~\ref{lem:momentreduct}:

\momentreduct*

Before proving the general reduction, we first make a much weaker claim:
\begin{lemma}
\label{lem:weakbound}
There exist constants $C_1$, $C_2$, $C_3$ such that for any $p < 2$, for any matrix $\AA$, there
exists a matrix $\AA'$ with at most $C_1 d^2$ rows such that
\begin{equation*}
\AA'^T {\WWbar'}^{1-2/p} \AA' \succeq \AA^T \WWbar^{1-2/p} \AA
\end{equation*}
and for all $\xx$,
\begin{equation*}
\norm{\AA' \xx}_p^p \preceq C_2 \norm{\AA \xx}_p^p.
\end{equation*}
Furthermore the Lewis weight of each row of $\AA'$ is at most $\frac{C_3}{d}$.
\end{lemma}
Note that we are only requiring the \emph{existence} of such an approximation, and that it is of size $\theta(d^2)$.
Also, we only will need this result to reduce the number of rows (and not the Lewis weight upper bound) subject to
random splitting; arguments that do not depend on the number of rows, such as that in~\cite{Talagrand90}, can simply
use split up versions of $\AA$.

\begin{proof}
Such an $\AA'$ simply obtained by sampling rows from $\AA$ proportional to Lewis weights, then scaling the result up
by a constant; then, with appropriate setting of constants, both conditions will hold with high probability.
The first condition follows from matrix Chernoff bounds plus Lemma~\ref{lem:fullstable}.  The second can be shown
with a simple union bound argument over a net, and was first shown in \cite{Schechtman}, Proposition 4.
\end{proof}

Now, we prove the actual reduction.
\begin{proof}[Proof of Lemma~\ref{lem:momentreduct}]
The argument proceeds by a standard symmetrization argument, as is used, for example, in the $\ell_2$ case in
\cite{RudelsonV07}.

For convenience, we let
\begin{equation*}
U = g(p, N + C_1 d^2, d, \epsilon / C_2, \delta)
\end{equation*}
and
\begin{equation*}
M = \expec{\SS}{\left ( \max_{\norm{\AA \xx}_p = 1} | \norm{\SS \AA \xx}_p^p - 1 | \right )^l}.
\end{equation*}

First, consider the functional taking $l$th power of the maximum absolute value taken by a function
(the expectation is of a quantity of this form).  This functional is convex, and
$\norm{\SS \AA \xx}_p^p - 1$ has mean 0.  Then our expectation satisfies
\begin{equation*}
M \leq \expec{\SS,\SS'}{\left ( \max_{\norm{\AA \xx}_p = 1} | \norm{\SS \AA \xx}_p^p - \norm{\SS' \AA \xx}_p^p | \right )^l}
\end{equation*}
where $\SS$ and $\SS'$ are two independent copies of the sampling process--this holds because subtracting the second copy is
adding a mean 0 random variable, which can only increase the expectation of any convex function.

We refer to the indices of the specific rows chosen for $\SS$ as $i_k$ (with $i'_k$ for $\SS'$).

We may explicitly write the inner expression here as a sum:
\begin{equation*}
\norm{\SS \AA \xx}_p^p - \norm{\SS' \AA \xx}_p^p = \left ( \sum_{k=1}^N \frac{|\aa_{i_k}^T \xx|^p}{\pp_{i_k}} \right ) - \left ( \sum_{k=1}^N \frac{|\aa_{i'_k}^T \xx|^p}{\pp_{i'_k}} \right ).
\end{equation*}

Since $i_k$ and $i'_k$ have identical distributions, we may independently randomly swap each of these pairs according to a random sign variable $\sigma_k$ without changing the distribution.
This swap would cause the first row to be subtracted and the second added, rather than vice versa
That random process is
\begin{equation*}
\left ( \sum_{k=1}^N \sigma_k \frac{|\aa_{i_k}^T \xx|^p}{\pp_{i_k}} \right ) - \left ( \sum_{k=1}^N \sigma_k \frac{|\aa_{i'_k}^T \xx|^p}{\pp_{i'_k}} \right ).
\end{equation*}

This is a sum of two (non-independent) copies of the same process: $\sum_{k=1}^N \sigma_k \frac{|\aa_{i_k}^T \xx|^p}{\pp_{i_k}}$.
Having two copies can only multiply an $l$-moment by $2^l$.  Thus we have
\begin{equation*}
M \leq 2^l \expec{i, \sigma}{\left ( \max_{\norm{\AA \xx}_1 = 1} \left | \sum_{k=1}^N \sigma_k \frac{|\aa_{i_k}^T \xx|^p}{\pp_{i_k}} \right | \right )^l}.
\end{equation*}

Now, we may consider the expected value of this with the rows taken ($i_k$) fixed, varying the $\sigma_k$.

First, we apply Lemma~\ref{lem:weakbound}, getting such an $\AA'$ with $C_1 d^2$ rows.  Now, we consider that
adding an additional $\sigma_i | (\aa')_i^T \xx |^p$ term can only increase the energy.  We define $\AA''$ as
$\SS \AA$ with the rows of $\AA$ appended.  The expected value is then at most
\begin{equation*}
\expec{\sigma}{\left ( \max_{\norm{\AA \xx}_1 = 1} \left | \sum_{k=1}^N \sigma_k |(\aa'')_k^T \xx|^T \right | \right )^l}.
\end{equation*}
Lemma~\ref{lem:fullmonotone} implies that the Lewis weight
of each row of $\AA''$ from $\AA'$ is at most $\frac{C_3}{d}$ (since adding the other rows can't bring them down)
while the quadratic ${\AA''}^T \WWbar'' \AA'' \succeq \AA^T \WWbar^{1-2/p} \AA$.  The latter implies that the
Lewis weight of each row from $\AA$ is at most $\frac{1}{U}$ (since it is weighted to have
$(\yy^T (\AA^T \WWbar^{1-2/p} \AA)^{-1} \yy)^{p/2} = \frac{1}{U})$.  Thus, assuming that the Lewis weight
 upper bounds are larger than $O(1/d)$, each of these has a Lewis weight of at most $\frac{1}{U}$,
as will be needed. To modify the proof to remove the need for this assumption, one can add potentially
multiple, downscaled copies of $\AA'$ depending on the bound on the remaining rows.

For a particular set of rows taken defining a matrix $\SS$, we let
\begin{equation*}
F = \max_{\norm{\AA \xx}_p = 1} | \norm{\SS \AA \xx}_p^p - 1 |.
\end{equation*}

Then for the corresponding matrix $\AA''$, we have, for all $\xx$
\begin{equation*}
\norm{\AA'' \xx}_p^p \leq (1+C+F) \norm{\AA \xx}_p^p.
\end{equation*}

This means that
\begin{equation*}
\max_{\norm{\AA \xx}_1 = 1} \left | \sum_{k=1}^N \sigma_k |(\aa'')_k^T \xx|^T \right |
\leq (1+C+F) \max_{\norm{\AA'' \xx}_1 = 1} \left | \sum_{k=1}^N \sigma_k |(\aa'')_k^T \xx|^T \right |
\end{equation*}
so that, applying the given moment bound on the random sign process,
\begin{align*}
\expec{\sigma}{\left ( \max_{\norm{\AA \xx}_1 = 1} \left | \sum_{k=1}^N \sigma_k |(\aa'')_k^T \xx|^T \right | \right )^l} &\leq (1+C+F)^l \expec{\sigma}{\left ( \max_{\norm{\AA'' \xx}_1 = 1} \left | \sum_{k=1}^N \sigma_k |(\aa'')_k^T \xx|^T \right | \right )^l} \\
&\leq (1+C+F)^l \epsilon^l \delta.
\end{align*}

$(1+C+F)^l \leq 2^{l-1} ((1+C)^l + F^l)$.  By definition, the expected value of $F^l$ under a random choice of $\SS$ is precisely $M$, the moment we are trying to bound.  We thus have
\begin{align*}
M &\leq 2^{2l-1} ((1+C)^l + M) \epsilon^l \delta \\
M &\leq \frac{2^{2l-1} (1+C)^l \epsilon^l \delta}{1 - 2^{2l-1} \epsilon^l \delta}.
\end{align*}

Then for some $\epsilon = O(1)$, the denominator is at least $\frac{1}{2}$, and $M \leq ((4+4C) \epsilon)^l \delta$.  Thus for sufficiently small $\epsilon$, obtaining the result for the random sign process with $\frac{\epsilon}{4+4C}$ gives what is needed.
\end{proof}

\end{appendix}

\end{document}